\documentclass[12pt]{article}

\usepackage{latexsym}
\usepackage{epsfig,graphics}
\usepackage{amsmath,amssymb,amsthm}
\usepackage{graphics,epsfig,calc}
\usepackage{upref}
 \newcommand{\beqn}{\begin{eqnarray}}
 \newcommand{\eeqn}{\end{eqnarray}}
 \newcommand{\be}{\begin{equation}}
 \newcommand{\ee}{\end{equation}}
 \newcommand{\ba}{\begin{array}}
 \newcommand{\ea}{\end{array}}

 \newcommand{\pa}{\partial}
 \newcommand{\re}{\ref}
 \newcommand{\ci}{\cite}
 \newcommand{\ds}{\displaystyle}
 \newcommand{\la}{\label}
 \newcommand{\rIm}{{\rm Im\5}}
 
 \newcommand{\fr}{\frac}

\newcommand{\ti}{\tilde}

\newcommand{\ve}{\varepsilon}
\newcommand{\vp}{\varphi}

\newcommand{\de}{\delta}

\newcommand{\cD}{{\cal D}}

\newcommand{\cH}{{\cal H}}

\newcommand{\cL}{{\cal L}}

\newcommand{\cM}{{\cal M}}
\newcommand{\cO}{{\cal O}}
\newcommand{\cP}{{\cal P}}

\newcommand{\al}{\alpha}
\newcommand{\ga}{\gamma}

\newcommand{\si}{\sigma}
\newcommand{\om}{\omega}

\newcommand{\na}{\nabla}
\newcommand{\Si}{\Sigma}
\newcommand{\lam}{\lambda}
\newcommand{\Lam}{\Lambda}

\newcommand{\5}{{\hspace{0.5mm}}}

\newcommand{\R}{\mathbb{R}}

\newcommand{\C}{\mathbb{C}}

\renewcommand{\theequation}{\thesection.\arabic{equation}}
\renewcommand{\thesection}{\arabic{section}}
\renewcommand{\thesubsection}{\arabic{section}.\arabic{subsection}}
\newtheorem{theorem}{Theorem}[section]

\renewcommand{\thetheorem}{\arabic{section}.\arabic{theorem}}
\newtheorem{defin}[theorem]{Definition}

\newtheorem{lemma}[theorem]{Lemma}
\newtheorem{remark}[theorem]{Remark}

\newtheorem{cor}[theorem]{Corollary}
\newtheorem{pro}[theorem]{Proposition}

\newcommand{\supp}{\mathop{\rm supp}\nolimits}

\newcommand{\bd}{\begin{defin}}
 \newcommand{\ed}{\end{defin}}
\newcommand{\bt}{\begin{theorem}}
 \newcommand{\et}{\end{theorem}}
\newcommand{\bp}{\begin{pro}}
 \newcommand{\ep}{\end{pro}}
\newcommand{\bl}{\begin{lemma}}
\newcommand{\el}{\end{lemma}}
\newcommand{\bc}{\begin{cor}}
\newcommand{\ec}{\end{cor}}
\newcommand{\br}{\begin{remark} }
\newcommand{\er}{\end{remark}}

\begin{document}
\begin{titlepage}

\begin{center}
{\Large\bf
Weighted  decay for \bigskip\\
magnetic  Schr\"odinger  equation}
\vspace{1cm}
\\
{\large A.~I.~Komech}
\footnote{
Supported partly by the Alexander von Humboldt
Research Award and by the Austrian Science Fund (FWF): P22198-N13
}\\
{\it Faculty of Mathematics Vienna University\\
and Institute for Information Transmission Problems RAS}\\
 e-mail:~alexander.komech@univie.ac.at
\medskip\\

{\large E.~A.~Kopylova}
\footnote{Supported partly by the 
Austrian Science Fund (FWF): M1329-N13
}\\
{\it Faculty of Mathematics Vienna University \\
and Institute for Information Transmission Problems RAS}\\
e-mail:~elena.kopylova@univie.ac.at
\end{center}

\date{}

\vspace{0.5cm}
\begin{abstract}
\noindent
We obtain a dispersive long-time decay in weighted 
norms for solutions of  3D   Schr\"odinger equation with generic
magnetic and scalar potentials.
The decay extends the results obtained by Jensen and
Kato for the  Schr\"odinger equation without magentic potentials.
For the proof we develop the spectral theory of Agmon, Jensen and Kato,
extending the high energy decay of the resolvent to the magnetic
Schr\"odinger equation.
\smallskip

\noindent
{\em Keywords}:  long-time decay,  weighted 
norms, magnetic  Schr\"odinger equation,
high energy decay of resolvent.
\smallskip

\noindent
{\em 2000 Mathematics Subject Classification}: 35L10, 34L25, 47A40, 81U05
\end{abstract}

\end{titlepage}

\setcounter{equation}{0}
\section{Introduction}
We establish a dispersive long time decay
for the solutions to 3D  magnetic  Schr\"odinger equation
\be\la{SE}
  i\dot\psi(x,t)=H\psi(x,t):=[-i\na-A(x)]^2\psi(x,t)+V(x)\psi(x,t),
\quad x\in\R^3
\ee
in weighted  norms.
For $s,\si\in\R$,
denote by $\cH^s_\si=\cH^s_\si (\R^3)$
the weighted Sobolev spaces introduced by Agmon, \ci{A},
with the finite norms
\be\la{norm}
  \Vert\psi\Vert_{\cH^s_\si}=\Vert\langle x
\rangle^\si\langle\na\rangle^s\psi\Vert_{L^2(\R^3)}<\infty,
\quad\quad \langle x\rangle=(1+|x|^2)^{1/2}.
\ee
We will also denote $L^2_\si=\cH^0_\si$.
We assume that $V(x)\in C(\R^3)$, $A_j\in C^2(\R^3)$ are
real  functions, and for some $\beta>3$ and $\beta_1>2$ the bounds hold
\be \label{V}
  |V(x)|+|A(x)|+|\na A(x)|
  \le C\langle x\rangle^{-\beta},\quad x\in\R^3,
\ee
\be \label{V1}
|\na\na A(x)|\le C \langle x\rangle^{-\beta_1},\quad x\in\R^3.
\ee
We restrict ourselves to the ``regular case'' in the terminology of
\cite{jeka} (or ``nonsingular case'' in  \cite{M}),
where the truncated resolvent of the
 operator $H$  is bounded at the edge point
$\lam=0$  of the continuous spectrum. In other words, the point
$\lam=0$ is neither eigenvalue nor resonance for the operator $H$;
this holds for {\it generic potentials}.

Our main result is the following long time decay of the solutions
to (\re{SE}): in the regular case,
\be\label{full}
 \Vert\cP_c\psi(t)\Vert_
 {\cH^0_{-\si}}=\cO (|t|^{-3/2}),\quad t\to\pm\infty
\ee
for initial data $\psi_0=\psi(0)\in\cH^0_\si$ with $\sigma>5/2$
where  $\cP_c$
is a Riesz projection onto the continuous spectrum of  $H$.
The decay is desirable for the study of asymptotic stability and scattering
for the solutions to nonlinear  equations.

Let us comment on previous results in this direction.
Asymptotic completeness for the magnetic
Schr\"odinger equation follows by methods of  the Birman-Kato theory \ci{RS}.
Spectral representation for this case has been obtained by Iwatsuka \ci{Iw82}
developing the Kuroda approach \ci{Kur1,Kur2}.
The Strichartz estimates for magnetic Schr\"odinger
equation with small potentials were obtained in
\ci{AF2008, GST}
and with large potentials in \ci{EGS}.
The decay in weighted  norms
has been obtained first by Jensen and Kato for the Schr\"odinger equation
with scalar potential \ci{jeka}.
However, for the magnetic Schr\"odinger equation
the decay in weighted  norms was not obtain up to now.

Finally, let us comment on our approach. We extend  methods of Agmon \ci{A},
and Jensen and Kato \ci{jeka},
to the magnetic Schr\"odinger equation.
Our main novelties - Theorems \re{AJK} and \re{HED} on
high energy decay for the magnetic resolvent, and
Lemmas \re{Agm2} and \re{Agm3} which are extensions of known Agmon's
Lemma A.2 and A.3 from \ci{A} (see also Lemma 4 from \ci[p. 442]{RS}).
Main problem in this extension - presence of the first order derivatives in the
perturbation. These derivatives cannot be handle with the perturbation theory
like \ci{A,jeka} since the correponding terms do not decay in suitable norms.
To avoid the perturbation approach, we apply spectral resolution for
magnetic Schr\"odinger operator in our extension
of  Lemma A.3 from \ci{A}. 

Our techniques rely on the D'Ancona-Fanelly
magnetic version of the Hardy inequality \ci{DF},  
spectral resolution established by Iwatsuka \ci{Iw82}, 
result of  Ionescu and Schlag  \ci{IS06}
on absence of singular spectrum, 
and result of  Koch and Tataru \ci{KT}
on absence of embedded eigenvalues in continuous spectrum.
We also apply limiting absorption principle
for the magnetic Schr\"odinger equation. We deduce the principle
by a suitable generalization of methods
of Agmon \ci{A}.

\setcounter{equation}{0}
\section{Free  equation}\la{fSE}
Denote by $R_0(\om)=(H_0-\om)^{-1}$ the resolvent of
the free Schr\"odinger operator  $H_0=-\Delta$.
The   resolvent is an
integral operator with the integral kernel
\be\la{ef}
R_0(\om,x-y)=\exp(i\om^{1/2}|x-y|)/4\pi|x-y|,
\quad\om\in\C\setminus[0,\infty),\quad\rIm\om^{1/2}>0.
\ee
Denote by $\cL (B_1,B_2)$ the Banach space of bounded linear operators
from a Banach space $B_1$ to a Banach space $B_2$.

Explicit formula (\re{ef}) implies the properties of $R_0(\om)$
which are obtained in \ci{A,jeka} (see also \cite[Appendix A]{3Dkg}):
\medskip\\
i) $R_0(\om)$ is  analytic function of $\om\in\C\setminus [0,\infty)$
with the values in $\cL(\cH^{m}_0,\cH^{m+2}_0)$ for any $m\in\R$;
\\
ii) The limiting absorption principle holds:
\be\la{lap}
R_0(\lam\pm i\ve)\to R_0(\lam\pm i0),\quad \ve\to 0+,\quad\lam>0
\ee
where the convergence holds
in $\cL(\cH^{m}_\si,\cH^{m+2}_{-\si})$ with  $\si>1/2$;
\\
iii) The asymptotics  hold for $\om\in\C\setminus [0,\infty)$,
\beqn\la{exp0}
\Vert R_0(\om)-R_0(0)\Vert_{\cL(\cH^{m}_\si,\cH^{m+2}_{-\si'})}
&\to& 0,\quad\om\to 0,\quad\si,\si'>1/2,\quad\si+\si'>2\\
\la{dif0}
\Vert R_0^{(k)}(\om)\Vert_{\cL(\cH^{m}_\si,\cH^{m+2}_{-\si})}
&=&\cO(\om^{1/2-k}),~~\om\to 0,~~\si>1/2+k,~~k=1,~2,...
\eeqn
iv) For $m\in\R$, $k=1,2,...$ and $\si>k+1/2$  the asymptotics hold
\be\la{A0k}
\Vert R_0^{(k)}(\om)\Vert_{\cL(\cH^m_\si,\cH^{m+l}_{-\si})}
=\cO(|\om|^{-\fr {1-l+k}2}),\quad \om\to\infty,~~
\om\in\C\setminus[0,\infty)
\ee
where $l=-1,0,1,2$ for $k=0$, and  $l=-1,0,1$ for $k=1,2,...$.
\medskip\\
Denote by $U _0(t)$ the dynamical group of the free  Schr\"odinger equation.
For $t\in\R$ and $\psi_0\in\cH^0_\si$ with $\si>1$, the group
$U _0(t)$ admits the integral representation
\be\la{int}
  U _0(t)\psi_0=\frac 1{2\pi i}\int\limits_0^{\infty} e^{-i\om t}
  \big[R _0(\om+i0)-R _0(\om-i0)\big]\Psi_0~ d\om
\ee
where the integral converges in the sense of distributions of $t\in\R$
with the values in $\cH^0_{-\si}$.

Representation (\re{int}),  properties i)-iii), and
bounds (\re{A0k}) imply the weighted  decay  (\ref{full})
for the solutions to the free  Schr\"odinger equation (see \cite{jeka}).
\setcounter{equation}{0}
\section{Perturbed  equation}\la{pSE}
\subsection{Limiting absorption principle}
Here we extend (\re{lap}) to perturbed resolvent
$R(\om)=(H-\om)^{-1}$, where
\be\la{HW}
H=H_0+W, ~~~~~~~
W\psi=(i\na\cdot{\bf A}+{\bf A}^2+V)\psi+i{\bf A}\cdot\na\psi.
\ee
\begin{theorem}\la{LAP}
Let condition (\re{V})  hold with  $\beta>2$.
Then
for $\lam>0$, the convergence  holds
\be\la{lapp}
R(\lam\pm i\ve)\to R(\lam\pm i0),\quad \ve\to 0+
\ee
in $\cL(L^2_\si, L^2_{-\si})$  with $\si>1/2$.
\end{theorem}
Theorem \re{LAP} follows from the result of \cite{EGS} where
slightly weaker conditions on potentials are imposed. For
the convenience of readers we give the independent proof in our
case. 
For the proof we will use the Born splitting \be\la{for}
R(\om)=[1+R_0(\om)W]^{-1}R_0(\om),\quad\om\in\C\setminus
[0,\infty) \ee where the operator function $[1+R_0(\om)W]^{-1}$ is
meromorphic in $\C\setminus [0,\infty)$ by the Gohberg-Bleher
theorem. The key role in the proofs of (\re{lapp}) plays the result
on the absence of the embedded
eigenvalues  in the continuous spectrum
which is known due from the paper of Koch and Tataru \ci{KT}.
We start with the following lemma
\begin{lemma}\la{lA}
i) Let  condition (\re{V}) hold with some $\beta >1$.
Then for $\lam > 0$, the operators
$R_0(\lam \pm i0)W:L^2_{-\si}\to L^2_{-\si}$
and $WR_0(\lam \pm i0):L^2_{\si}\to L^2_{\si}$
are compact for $\si\in(1/2,\beta-1/2)$.
\\
ii)
Let  condition (\re{V}) hold with some $\beta >2$.
Then the operators
$R_0(0)W:L^2_{-\si}\to L^2_{-\si}$
and $WR_0(0):L^2_{\si}\to L^2_{\si}$
are compact for $\si\in(1/2,\beta-1/2)$.
\end{lemma}
\begin{proof}
i)
Choose $\si'\in(1/2,\min(\si,\beta-\si))$. The operator
$W: L^2_{-\si}\to \cH^{-1}_{\si'}$ is continuous by (\re{V}) since
$\si+\si'<\beta$.
Further,
$R_0(\lam \pm i0):  \cH^{-1}_{\si'}\to \cH^1_{-\si'}$
is continuous by (\re{lap})
and the embedding $\cH^1_{-\si'}\to L^2_{-\si}$
is compact by the Sobolev embedding theorem.
Hence, the operators
$R_0(\lam \pm i0)W: L^2_{-\si}\to L^2_{-\si}$
are compact. The compactness of
$WR_0(\lam \pm i0): L^2_{\si}\to L^2_{\si}$ follows by duality.
\\
ii)
Choose sufficienlty small $\ve>0$ such that
\be\la{sie}
  \si':=\beta-\si-\ve>1/2,~~~~~~~\si'':=\si-\ve>1/2,~~~~~~~
  \si'+\si''=\beta-2\ve>2.
\ee
The operator $W: L^2_{-\si}\to \cH^{-1}_{\si'}$ is continuous by (\re{V}) since
$\si'+\si=\beta-\ve<\beta$, and the operator
$R_0(0): \cH^{-1}_{\si'}\to \cH^1_{-\si''}$
is continuous by (\re{exp0}).
The embedding $\cH^1_{-\si''}\to L^2_{-\si}$
is compact by the Sobolev embedding theorem.
Hence, the operator
$R_0(0)W: L^2_{-\si}\to L^2_{-\si}$
is compact. The compactness of
$WR_0(0):L^2_{\si}\to L^2_{\si}$ follows by duality.
\end{proof}
Theorem \re{LAP} will follow from  convergence (\re{lap}) and the
Born splitting (\re{for}) if
$$
  [1+R_0(\lam\pm i\ve)W]^{-1}\to[1+R_0(\lam\pm i0)W]^{-1},
  \quad\ve\to +0,\quad\lam>0
$$
in $\cL(L^2_{-\si}, L^2_{-\si})$  with $\si>1/2$.
The convergence holds if and only if the both limiting operators
$1+R_0(\lam \pm i0)W:L^2_{-\si}\to L^2_{-\si}$
are invertible for $\lam>0$.
The operators are invertible according to the Fredholm theorem
by Lemma \re{lA} i) and the following lemma.
\begin{lemma}\la{pAK}
Let condition (\re{V}) holds with some $\beta>2$. Then
for $\lam>0$ the equations
\be\la{AK}
  [1+R_0(\lam \pm i0)W]\psi=0
\ee
admit only zero solution in $L^2_{-1/2-0}$.
\end{lemma}
\begin{proof}
We adopt general strategy from \ci{A}.
\medskip\\
 {\it Step i)} We consider the case $\lam +i0$ for concreteness.
Equality (\re{AK}) implies that
\be\la{eqi}
  (H-\lam)\psi=(H_0-\lam)(1+R_0(\lam+ i0)W)\psi=0.
\ee
We will show that if  $\psi\in L^2_{-1/2-0}$ is the solution to (\re{AK})
then $\psi\in L^2$, i.e
$\psi$ is the eigenfunction of
$H$ corresponding to the positive eigenvalue $\lam>0$.
However, the embedded eigenvalue is forbidden,
and  then $\psi=0$.
\medskip\\
 {\it Step ii)}
From  (\re{AK}) it follows that
\be\la{eqii}
  \psi=R_0(\lam+ i0)f,~~{\rm where}~~f=-W\psi.
\ee
Moreover, (\re{AK})  also implies  that $\psi\in \cH^1_{-1/2-0}$.
Hence,   $f\in L^2_{1}$ by (\re{V}) with $\beta>3/2$.
In the Fourier transform, equation (\re{eqii}) becomes
$$
 \hat\psi(\xi)=\fr{\hat f(\xi)}{\xi^2-\lam-i0},~~~~~\xi\in\R^3
$$
where $\hat f$ is a function from the Sobolev space $\cH^{1}$.
\medskip\\
{\it Step iii)}
Next, we prove  that
\be\la{fS0}
  \hat f(\xi)\Big|_{S_{\sqrt\lam}}=0
\ee
where
$S_{\sqrt\lam}:=\{\xi\in\R^3:|\xi|=\sqrt{\lam}\}$.
Note that the trace on the sphere $S_{\sqrt\lam}$ exists, and
$\hat f(\xi)\Big|_{S_{\sqrt\lam}}\in\cH^{1/2}(S_{\sqrt\lam})$.
Moreover, in the polar coordinates
$r=|\xi|\in [0,\infty)$, $\vp=\xi/|\xi|\in S_1$, the map
$$
  M: [0,\infty)\to L^2(S_1),\quad M(r) = \hat f(r\vp),~~~\vp\in S_1
$$
is H\"older continuous with the H\"older exponent $\al\in(0,1/2)$.
This follows from the Sobolev theorem on the traces (\ci[Ch. 1]{LM}).
Define
$$
  \hat\psi_\ve(\xi)=\fr{\hat f(\xi)}{\xi^2-\lam-i\ve},~~~~\ve>0.
$$
Then both $\hat f, \hat\psi_\ve\in L^2(\R^3)$, hence
the Parseval identity implies that
\beqn\la{fVpie}
  \!\!\!\!\!\!\!\!&&(\psi_\ve,f)=(\hat\psi_\ve,\hat f)\nonumber\\
  \!\!\!\!\!\!\!\!&&=\int \fr{|\hat f(\xi)|^2}{\xi^2-\lam-i\ve} d\xi
  \longrightarrow \fr{i\pi}{2\sqrt\lam}
  \int\limits_{S_{\sqrt\lam}}|\hat f(\xi)|^2dS(\xi)+
  \lim_{\de\to 0}\int\limits_{||\xi|-\sqrt\lam|>\de}
  \fr{|\hat f(\xi)|^2}{\xi^2-\lam}d\xi,~~~\ve\to 0+
\eeqn
by the  Sokhotsky-Plemelj formula since the map
$$
M_1: [0,\infty)\to L^1(S_1),\quad M_1(r) = |\hat f(r\vp)|^2,~~~\vp\in S_1
$$
is the H\"older continuous with the same H\"older exponent $\al\in(0,1/2)$.
On the other hand,
\be\la{oh}
  (\psi_\ve,f)=(R_0(\lam+i\ve)f,f)\longrightarrow (\psi,f)
  =-(\psi,W\psi),~~~~~\ve\to 0+
\ee
since $R_0(\lam+i\ve)f\to\psi$ in $\cH^2_{-1/2-0}$ by (\re{lap}),
while $f\in L^2_1$.
The operator $W$ is selfajoint,
hence the scalar product $(\psi,W\psi)$ is real.
Comparing (\re{fVpie}) and (\re{oh}), we conclude  that
$$
\int_{S_{\sqrt\lam}}|\hat f(\xi)|^2dS(\xi)=0
$$
i.e. (\ref{fS0}) is proved.
Relation (\re{fS0}) and the H\"older continuity imply that
\be\la{hp}
 \hat\psi(\xi)=\fr{\hat f(\xi)}{\xi^2-\lam}\in L^1_{\rm loc}(\R^3).
\ee
{\it Step iv)}
Finally we prove that
\be\la{fHs}
\Vert\hat\psi\Vert_{L^2}\le C\Vert \hat f\Vert_{\cH^1}.
\ee
For the proof we take any $\ve\in (0,\sqrt\lam/2)$, and a cutoff function
$$
  \zeta(\xi)\in C_0^\infty(\R^3),\quad\zeta(\xi)=\left\{\ba{ll}
  1,& ||\xi|-\sqrt\lam|<\ve\\
  0,& ||\xi|-\sqrt\lam|>2\ve
  \ea\right.
$$
By (\re{hp}), we have
$$
  \Vert(1-\zeta(\xi))\hat\psi(\xi)\Vert_{\cH^1}
  =\Vert\fr{1-\zeta(\xi)}{\xi^2-\lam}\hat f(\xi)\Vert_{\cH^1}
  \le C\Vert \hat f\Vert_{\cH^1}\,.
$$
Hence, it remains to estimate the norm of the function
$\zeta(\xi)\hat\psi(\xi)$. Choose a finite partition of unity
$\sum\zeta_j(\xi)=1$, $\xi\in\supp\zeta$,
with $\zeta_j\in C_0^\infty(\R^3\setminus 0)$.
We may assume that in the $\supp\zeta_j$, for every fixed $j$, there exist
the  corresponding local coordinates
$\eta_1,\eta_2,\eta_3$ with $\eta_1=\xi^2-\lam$.
Then, the problem reduces to the estimate
$$
\Vert\vp(\eta)\Vert_{L^2}\le C\Vert \eta_1\vp(\eta)\Vert_{\cH^1}
$$
taking into account that $\vp(\eta)\in L^1(\R^3)$ by (\re{hp}).
It suffices to prove the bound
\be\la{fHser}
\Vert\phi(x)\Vert_{L^2}\le C\Vert \pa_1\phi(x)\Vert_{L^2_1}
\ee
for the function $\phi(x):=F^{-1}\vp$, knowing that
$\phi(x)\to 0$ as $|x|\to\infty$
by the Riemann-Lebesgue theorem.
Bound (\re{fHser}) follows by the Hardy inequality (see \ci{A})
$$
\int\phi^2(x_1,x')dx_1\le 4 \int |x_1|^{2}
|\pa_1\phi(x_1,x')|^2dx_1,\quad{\rm a.a.}~~x':=(x_2,x_3)\in\R^2
$$
by integration over $x'\in\R^2$.
Now (\re{fHs}) is proved.
Finally, (\re{fHs}) can be rewritten as
$$
\Vert\psi\Vert_{L^2}\le C\Vert f\Vert_{L^2_1}
$$
that proves Lemma \re{pAK}.
\end{proof}
Now the proof of  Theorem \re{LAP} is also completed.

\begin{cor}\la{cLAP}
Under the conditions of Theorem \re{LAP},  convergence (\re{lapp})
holds in  $\cL(\cH^0_{\si}, \cH^2_{-\si})$ with $\si>1/2$.
\end{cor}
\begin{proof}
The operators $1+WR_0(\lam\pm i0): L^2_{\si}\to L^2_{\si}$
are adjoint to $1+R_0(\lam\mp i0)W: L^2_{-\si}\to L^2_{-\si}$.
The operators $1+R_0(\lam\mp i0)W$ are invertible by Lemma \re{pAK},
hence $1+WR_0(\lam\pm i0)$ also are invertible by the Fredholm theorem.
Therefore, the corollary follows by the alternative
Born splitting
\be\la{fort2}
R(\om)=R_0(\om)[1+WR_0(\om)]^{-1},\quad\om\in\C\setminus [0,\infty)
\ee
and convergence (\re{lap}).
\end{proof}
\subsection{Zero point $\om=0$}
Here we consider  $R(\om)$ near  $\om=0$.
We set
$$
  \cM=\{\psi\in L^2_{-1/2-0}: \psi+R_0(0)W\psi=0\}.
$$
The functions $\psi\in \cM\cap L^2$ are the zero eigenfunctions of
$H$ since $H\psi=H_0(1+R_0(0)W)\psi=0$
by splitting (\re{for}).
The functions $\psi\in \cM\setminus L^2$ are called  the {\bf zero
resonances} of  $H$.

Our key assumption is (cf. Condition (i) in \ci[Theorem 7.2]{M}):
\be\la{SC}
{\bf Spectral~ Condition:}
~~~~~~~~~~~~~~~~~~~\cM=0~~~
~~~~~~~~~~~~~~~~~~~~~~~~~~~~
\ee
In the other words, the point $\om=0$ is neither eigenvalue nor resonance for
the operator $H$. Condition (\re{SC})
holds for {\it generic $W$}.
\begin{lemma}\la{lp0}
Let condition (\re{V}) with a $\beta> 2$ and
Spectral Condition (\re{SC}) hold. Then the discrete spectral set
$\Si$ is finite, and for $\si,\si'>1/2$
with $\si+\si'>2$, the  asymptotics hold,
\be\la{lappexp}
  \Vert R(\om)-R(0)\Vert_{\cL(\cH^0_{\si},\cH^2_{-\si'})}
  \to 0,\quad\om\to 0,\quad\om\in \C\setminus[0,\infty)
\ee
where the operator $R(0):\cH^0_{\si}\to \cH^2_{-\si'}$ is continuous.
\end{lemma}
\begin{proof}
It suffices to consider the case
$$
1/2<\si,\si'<\beta-1/2,~~\si+\si'>2
$$
since asymptotics (\re{lappexp}) hold then for larger $\si,\si'$.
Spectral Condition (\re{SC}) implies that the operators
$[1+R_0(0)W]:L^2_{-\si}\to L^2_{-\si}$
and  $[1+WR_0(0)]:L^2_{\si}\to L^2_{\si}$
are invertible by Lemma \re{lA} ii) and the Fredholm theorem.
Then the operator  $[1+WR_0(\om)]:L^2_{\si}\to L^2_{\si}$
also is invertible
and the operator function $[1+WR_0(\om)]^{-1}$
with the values in $\cL(L^2_{\si}, L^2_{\si})$ is continuous
for small $\om\in \C\setminus[0,\infty)$. Therefore,
convergence (\re{lappexp}) holds by (\re{fort2}) and (\re{exp0}).
\end{proof}

\begin{lemma}\la{lZc}
Let condition (\re{V}) with a $\beta>3$ and
Spectral Condition (\re{SC}) hold. Then
\beqn
\!\!\!\!\!\!\!\!\!
\Vert R(\om)\Vert_{\cL(\cH^0_{\si},\cH^2_{-\si})}\!\!\!&=&\!\!\!\cO(1),
\quad\om\to 0,\quad \om\in \C\setminus[0,\infty),\quad \si>1,
\la{Zcaexp}
\\
\!\!\!\!\!\!\!\!\!\Vert R^{(k)}(\om)\Vert_{\cL(\cH^0_{\si},\cH^2_{-\si})}
\!\!\!&=&\!\!\!\cO(|\om|^{1/2-k}),~~\om\!\to 0,~~\om\in \C\setminus[0,\infty),
~~\si>1/2+k,~~k=1,2. \la{Zca}
\eeqn
\end{lemma}
\begin{proof}
Bound (\re{Zcaexp}) holds by Lemma \re{lp0}.
To prove  (\re{Zca}) with $k= 1$ we apply the identity
\be\la{R1ml}
  R'=(1-RW)R_0'(1-WR)=R_0'-RWR_0'-R_0'WR+RWR_0'WR.
\ee
The relation implies (\re{Zca}) with $k=1$ and $\si>3/2$
by (\re{dif0}) with $k=1$ and (\re{Zcaexp}).
Namely, for the first term in the RHS of (\re{R1ml})
this is obvious. Consider the second term.
Choosing $\si'\in (3/2,\beta-3/2)$, we obtain
\be\la{fr}
  \Vert R(\om)WR_0'(\om)\psi\Vert_{\cH^2_{-\si}}
  \le C \Vert WR_0'(\om)\psi\Vert_{L^2_{\si'}}
  \le C_1\Vert R_0'(\om)\psi\Vert_{\cH^1_{\si'-\beta}}
  \le C_2|\om|^{-1/2}\Vert \psi\Vert_{L^2_{\si}}.
\ee
The remaining terms can be estimated similarly.
Hence, (\re{Zca}) with $k=1$ is proved.
\medskip\\
For $k=2$ we apply the formula:
\beqn\la{Rk21}
   R''&=&(1-RW)R_0''(1-WR)-2 R'WR_0'(1-WR)\\
   \nonumber
   &=&R_0''-RWR_0''-R_0''WR+RWR_0''WR-2 R'WR_0'+2 R'WR_0'WR.
\eeqn
Bound (\re{Zca}) with $k=2$ and $\si> 5/2$ for the first term
in the RHS of (\re{Rk21}) follows from (\re{dif0}) with  $k=2$.
The last two terms can be estimated similarly to
(\re{fr}) using (\re{Zcaexp}) and (\re{Zca})  with $k=1$.
Consider the remaining terms.
Using (\re{Zcaexp})  and (\re{dif0}) with $k=2$, we obtain that
\\
a) for $\si'\in(5/2,\beta-1/2)$ the bounds hold
\beqn\nonumber
\Vert RWR_0''\psi\Vert_{\cH^2_{-\si}}
&\le& C\Vert WR_0''\psi\Vert_{L^2_{-\si'+\beta}}
\le C_1\Vert R_0''\psi\Vert_{\cH^1_{-\si'}}
\le C_2|\om|^{-3/2}\Vert \psi\Vert_{L^2_{\si}},\\
\nonumber\\
\nonumber
\Vert R_0''WR\psi\Vert_{\cH^2_{-\si}}
&\le& C|\om|^{-3/2}\Vert WR\psi\Vert_{L^2_{\si'}}
\le C_1|\om|^{-3/2}\Vert R\psi\Vert_{\cH^1_{\si'-\beta}}
\le C_2|\om|^{-3/2}\Vert\psi\Vert_{L^2_{\si}}
\eeqn
by Lemma \re{lp0} since $-\si'+\beta>1/2$ and $\si+\beta-\si'>2$.
\\
b) for $\si'\in(1/2,\beta-5/2)$ the bound holds
\beqn\nonumber
\Vert RWR_0''WR\psi\Vert_{\cH^2_{-\si}}
&\le& C\Vert WR_0''WR\psi\Vert_{L^2_{\si'}}
\le C_1\Vert R_0''WR\psi\Vert_{\cH^1_{\si'-\beta}}\\
\nonumber\\
\nonumber
&\le& C_2|\om|^{-3/2}\Vert WR\psi\Vert_{\cL^2_{-\si'+\beta}}
\le C_3|\om|^{-3/2}\Vert R\psi\Vert_{\cH^1_{-\si'}}
\le C_4|\om|^{-3/2}\Vert\psi\Vert_{L^2_{\si}}
\eeqn
by  Lemma \re{lp0} since $\beta-\si'>5/2$ and $\si+\si'>2$.
Hence,  (\re{Zca}) with $k=2$  is proved.
\end{proof}
\subsection{ High energy decay}
Denote by $R_A(\om)=(H_A-\om)^{-1}$ the resolvent of
the operator $H_A=[-i\na-A(x)]^2$.
In Appendix A we will prove the folowing asymprotics  of $R_A(\om)$
for large $\om$:
\begin{theorem}\la{AJK}
Let $A_j(x)\in C^2(\R^3)$ are real functions and for some $\beta>2$
the bound holds
\be\la{nnA}
  |A(x)|+|\na A(x)|+|\na\na A(x)|\le C\langle x\rangle^{-\beta}
\ee
Then for $\si>1/2$ and
$l=0;1$ the asymptotics hold 
\be\la{A1}
 \Vert R_A(\om)\Vert_{{\cal L}(\cH^0_\si;\cH^{l}_{-\si})}
 ={\cal O}(|\om|^{-\fr{1-l}2}),
  \quad |\om|\to\infty,\quad\om\in\C\setminus[0,\infty).
\ee
\end{theorem}
Now we derive the asymptotics of $R(\om)$ and its derivatives
for large $\om$ from (\re{A1}).

\begin{theorem}\la{HED}
Let (\re{V})with $\beta>3$ and (\re{V1}) with $\beta>2$ hold.
 Then for $k=0,1,2$,
$\si>1/2+k$,  and $l=0,1$, the asymptotics hold \be\la{H}
 \Vert R^{(k)}(\om)\Vert_{{\cal L}(\cH^0_\si;\cH^{l}_{-\si})}
 ={\cal O}(|\om|^{-\fr{1-l+k}2}),
  \quad |\om|\to\infty,\quad\om\in\C\setminus[0,\infty).
\ee
\end{theorem}
\begin{proof}
{\it Step i)}
For $k=0$ asymptotics (\re{H}) follows from the  Born splitting
$$
 R(\om)=R_A(\om)[1+VR_A(\om)]^{-1}
$$
and (\re{A1}), since the norm of
$[1+VR_A(\om)]^{-1}: \cH^0_{\si}\to \cH^0_{\si}$ is bounded
for large $\om\in\C\setminus[0,\infty)$ and $\si\in(1/2,\beta/2]$.
\\
{\it Step ii)}
For $k=1$ we use identity (\re{R1ml}).
The identity implies  (\re{H}) with $k=1$
and $\si>3/2$ by (\re{A0k})  with $k=1$, and (\re{H}) with $k=0$.
Indeed, this is obvious
for the first term in the RHS of (\re{R1ml}).
Let us consider the second term. Choosing $\si'\in (3/2,\beta-3/2)$,
we obtain for large $\om\in\C\setminus [0,\infty)$
\be\la{3.34}
  \Vert RWR_0'\psi\Vert_{\cH^l_{-\si}}\le C|\om|^{-\fr{1-l}2}
  \Vert WR_0'\psi\Vert_{\cH^0_{\si'}}
  \le C_1|\om|^{-\fr{1-l}2}\Vert R_0'\psi\Vert_{\cH^{1}_{\si'-\beta}}
  \le C_2|\om|^{-\fr{2-l}2}\Vert \psi\Vert_{\cH^0_{\si}},~~l=0;1.
\ee
The remaining terms in the RHS of (\re{R1ml}) can be estimated similarly.
Hence, (\re{H}) with $k=1$ and $\si> 3/2$ is proved.
\medskip\\
{\it Step iii)}
In the case $k=2$ we apply identity (\re{Rk21}).
Asymptotics (\re{H}) with $k=2$
for the first term in the RHS of (\re{Rk21}) follows from
(\re{A0k}) with  $k=2$.
The last two terms can be estimated similarly to
(\re{3.34}) using (\re{A0k}) with  $k=1$ and (\re{H}) with $k=0;1$.

Consider the remaining terms.
Using (\re{H}) with $k=0$
and (\re{A0k}) with $k=2$ and $l=0;1$, we obtain that
\\
a) for $\si'\in(5/2,\beta-1/2)$ the bounds hold
\beqn\nonumber
 \Vert RWR_0''\psi\Vert_{\cH^l_{-\si}} &\le& C|\om|^{-\fr{1-l}2}
 \Vert WR_0''\psi\Vert_{\cH^0_{-\si'+\beta}}\le C_1|\om|^{-\fr{1-l}2}
 \Vert R_0''\psi\Vert_{\cH^{1}_{-\si'}}\le C_2|\om|^{-\fr{3-l}2}
 \Vert \psi\Vert_{\cH^0_{\si}},\\
\nonumber\\
\nonumber
\Vert R_0''WR\psi\Vert_{\cH^l_{-\si}}&\le& C|\om|^{-\fr{3-l}2}
\Vert WR\psi\Vert_{\cH^0_{\si'}}\le C_1|\om|^{-\fr{3-l}2}
\Vert R\psi\Vert_{\cH^{1}_{\si'-\beta}}\le C_2|\om|^{-\fr{3-l}2}
\Vert \psi\Vert_{\cH^0_{\si}},
\eeqn
b) for $\si'\in(1/2,\beta-5/2)$ the bound holds
\beqn\nonumber
 \Vert RWR_0''WR\psi\Vert_{\cH^l_{-\si}}&\le& C|\om|^{-\fr{1-l}2}
\Vert WR_0''WR\psi\Vert_{\cH^0_{\si'}}\le C_1|\om|^{-\fr{1-l}2}
\Vert R_0''WR\psi\Vert_{\cH^{1}_{\si'-\beta}}\\
\nonumber\\
\nonumber
&\le& C_2|\om|^{-\fr{3-l}2}\Vert WR\psi\Vert_{\cH^0_{-\si'+\beta}}
\le C_3|\om|^{-\fr{3-l}2}\Vert R\psi\Vert_{\cH^{1}_{-\si'}}
\le C_4|\om|^{-\fr{3-l}2}\Vert\psi\Vert_{\cH^0_{\si}}.
\eeqn
Hence,
(\re{H}) with $k=2$ is proved.
\end{proof}
\setcounter{equation}{0} \setcounter{theorem}{0}
\setcounter{equation}{0}
\section{Time  Decay}
We prove  time decay (\re{full}) follow the methods of \cite{jeka}.
Under  conditions (\re{V}), (\re{V1}) and  (\re{SC})
the solution $\psi(t)$ to (\re{SE}) admits the representation
\be\la{srelap}
\psi(t)=
\sum_{j=1}^N
 e^{-i\om_j t}P_j\psi(0)+\fr{1}{2\pi i}\int_0^\infty
e^{-i\om t}[R(\om+i0)-R(\om-i0)]\psi(0)~d\om,~~~t\in\R
\ee
for initial state $\psi(0)\in L^2_\si$ with $\si>1$.
The representation follows from  the Cauchy residue theorem,
Theorem \re{LAP}, and  (\re{H}) with $k=0$.
\begin{defin}\la{Xdc}
i) $X_d:=\sum_{j=1}^N P_j L^2$ is the discrete spectral subspace
of $H$ spanned by all eigenfunctions.
\\
ii)$X_c:=X_d^\bot$ is the  orthogonal to $X_d$ subspace of the
continuous spectrum of  $H$.
\end{defin}
\begin{theorem}\la{twed}
Let  conditions (\re{V}), (\re{V1}) and (\re{SC}) hold. Then
\be\la{wedp}
\Vert\psi(t)\Vert_{L^2_{-\si}}\le C\langle t\rangle^{-3/2}
\Vert \psi(0)\Vert_{L^2_\si}, ~~~~~~t\in\R
\ee
for any initial state
$\psi(0)\in X_c\cap L^2_\si$ with $\si>5/2$.
\end{theorem}
\begin{proof}
Since $P_j\psi(0)=0$, then (\re{srelap}) reduces to
\be\la{srelar}
\psi(t)=\fr{1}{2\pi i}\int_0^\infty
e^{-i\om t}[R(\om+i0)-R(\om-i0)]\psi(0)~d\om,~~~t\in\R.
\ee
To deduce (\re{wedp}), introduce the partition of unity
$1=\zeta_l(\om)+\zeta_h(\om)$, $\om\in\R$,
where
$$
  \zeta_l\in C_0^\infty(\R),~~~~~\zeta_l(\om)=\left\{\ba{ll}
  1,&|\om|\le \ve/2  \\ 0,&|\om|\ge\ve \ea\right.
$$
with a small $\ve>0$. Then (\re{srelar}) reads
\beqn\la{srelab}
  \psi(t)=\psi_l(t)+\psi_h(t)&=&
  \fr{1}{2\pi i}\int_0^\infty \zeta_l(\om)
  e^{-i\om t}[R(\om+i0)-R(\om-i0)]\psi(0)~d\om
  \nonumber\\
  \nonumber\\
  &+&
  \fr{1}{2\pi i}\int_0^\infty \zeta_h(\om)
  e^{-i\om t}[R(\om+i0)-R(\om-i0)]\psi(0)~d\om.
\eeqn
Integrating twice by parts and using  (\re{H}) with $k=2$
we obtain for the ``high energy component'' $\psi_h(t)$ the decay
$$
 \Vert\psi_h(t)\Vert_{L^2_{-\si}}\le C\langle t\rangle^{-2}
 \Vert\psi(0)\Vert_{L^2_\si}.
$$
To estimate  the ``low energy component'' $\psi_l(t)$
we  apply the following  lemma of Jensen-Kato
\cite [ Lemma 10.2]{jeka} to the vector function
$F(\om):=\zeta_l(\om)[R(\om+i0)-R(\om-i0)]\psi(0)$
with the values in  the Banach space  ${\bf B}=L^2_{-\si}$ with $\si>5/2$:
\begin{lemma}\la{jk}
Let $F\in C(0, a; {\bf B})$ satisfy
\be\la{Zc}
 F(0)=F(a)=0;~~~\Vert F''(\om)\Vert_{\bf B}\le C|\om|^{-3/2},~~\om\in(0,a).
\ee
Then
\be\la{Zyg}
  \Vert\int\limits_0^a e^{-i\omega t}F(\omega)d\omega\Vert_{\bf B}
  =\cO (t^{-3/2}),\quad t\to\infty.
\ee
\end{lemma}
All the conditions of Lemma \re{jk} are satisfied
due to (\re{Zcaexp})-(\re{Zca}). Then
\be\la{lec}
  \Vert\psi_l(t)\Vert_{\cL^2_{-\si}}\le C\langle t\rangle^{-3/2}
  \Vert\psi(0)\Vert_{L^2_\si}.
\ee
\end{proof}
\appendix

\setcounter{section}{0}
\setcounter{equation}{0}
\protect\renewcommand{\thesection}{\Alph{section}}
\protect\renewcommand{\theequation}{\thesection.\arabic{equation}}
\protect\renewcommand{\thesubsection}{\thesection.\arabic{subsection}}
\protect\renewcommand{\thetheorem}{\Alph{section}.\arabic{theorem}}
\section{Proof of Theorem \re{AJK}}
Here we extend the  Agmon-Jensen-Kato estimates \cite[(A.2')]{A} and
\cite[(8.1)]{jeka}
to the resolvent $R_A(\om)$.
The operator $H_A$ for $A(x)\in C^1(\R^3)$  is a symmetric
operator in the Hilbert space $L^2:=L^2(\R^3)$
with the domain $\cD:=C_0^\infty(\R^3)$. Moreover, $H_A$ is
nonnegative, hence it admits the unique
selfadjoint extension which is its closure,
by the Friedrichs theorem. Denote by
$H_A^{1/2}$  the nonnegative square root of $H_A$  which is also selfadjoint
operator in $L^2$, so
$$
  \Vert H_A^{1/2} u\Vert
  =\Vert\nabla_A u\Vert, ~~~~~~~~~~~u\in\cD
$$
where $\nabla_A=\nabla-iA$, and
$\Vert\cdot\Vert$ stands for the norm in $L^2$.
\begin{lemma}\la{eqv}
Let $A(x)\in C(\R^3)$ and $|A(x)|\le C\langle x\rangle^{-\beta}$
with a $\beta\ge 1$.
Then for any $\si\in\R$, the bounds hold
\be\la{eqv1}
   \Vert\nabla u\Vert\le C_1 \Vert\nabla_A u\Vert
   \le C_2 \Vert\nabla u\Vert,\quad u\in \cD.
\ee
\end{lemma}
\begin{proof}
We apply the magnetic version of  the Hardy inequality (\cite{DF}):
\be\la{M}
  \Vert u\Vert_{L^2_{-1}} \le
  4\Vert\nabla_A u\Vert, \quad u\in\cD.
\ee
 Writing
$ \nabla u = \bigl( \nabla -iA(x) \bigr) u + iA(x)u $,
we obtain by  (\ref{M}),
$$
 \Vert\nabla u\Vert
 \le\Vert\bigl(\nabla-iA(x)\bigr) u\Vert +\Vert Au\Vert
 \le\Vert\nabla_Au\Vert +C\Vert u\Vert_{L^2_{-1}}
 \le C_1\Vert\nabla_A u\Vert.
$$
Further,
$$
 \Vert\nabla_A u\Vert=\Vert\bigl(\nabla\!-\! iA(x)\bigr)
u\Vert
\le  \Vert\nabla u\Vert+\Vert Au\Vert
\le  \Vert\nabla u\Vert+C\Vert u\Vert_{L^2_{-1}}
\le C_2\Vert\nabla u\Vert
$$
where the last bound follows from
(\ref{M}) with $A(x)\equiv 0$.
\end{proof}
We reduce Theorem \ref{AJK} to  certain lemmas.
The first lemma generalizes Lemma A.2 from \ci{A}.
\begin{lemma}\la{Agm2} Let the conditions of Theorem \ref{AJK} hold.
Then for $\sigma>1/2$, the bound holds
 \be\la{A2}
   \Vert (-i\na-A) \psi \Vert_{\cH^0_{-\si}}
   \le C(\si)\Vert (H_A -\om)\psi \Vert_{\cH^0_\si},\quad \psi\in\cD,
   \quad\om\in\C.
 \ee
\end{lemma}
\begin{proof}
It suffices to estimate each component:
\be\la{A2j}
  \Vert (-i\na_j-A_j(x))  \psi \Vert_{\cH^0_{-\si}}
  \le C(\si)\Vert (H_A -\om)\psi \Vert_{\cH^0_\si},\quad j=1,2,3.
\ee
Consider $j=1$ for the concreteness.
Applying the gauge transformation $\psi(x)\mapsto\psi(x)e^{i\Phi(x)}$
with $\Phi(x)=\ds\int A_1(x)dx_1$, we reduce the estimate to the case
$A_1(x)=0$ and $A'_j(x)=A_j(x)-\na_j\Phi(x)$ instead of $A_j(x)$ for $j\ne 1$.
By (\re{nnA}), for 
the real functions $A'_j(x)$ the bound holds
\be\la{potdecg}
  |A'_j(x)|+|\na A'_j(x)|\le C\langle  x'\rangle^{-\beta},
  ~~~~~~x':=(x_2,x_3)\in\R^{2}.
\ee
Note that this is the only place
we need  the condition (\re{V1}) on the second derivatives of $A_j$.
Now (\re{A2j}) reduce to  bound
\be\la{A2ecrs}
  \int\langle x_1\rangle^{-2\si}|\na_1  \psi(x)|^2 dx
  \le C(\si)\int
  \langle x_1\rangle^{2\si} |(-\na_1^2 +\Lam(x_1)-\om)\psi(x)|^2dx.
\ee
where
$$
  \Lam(x_1)=\sum_2^3[-i\na_j-A'_j(x)]^2
$$
is a nonnegative selfadjoint operator in $L^2(\R^{2})$.
From \ci[Theorem 1.3 (e)]{IS06} it follows that $\cH_{sing}(\Lam(x_1))=0$.
Moreover,  $\cH_{pp}(\Lam(x_1))=0$ since
eigenvalues $\lam_j>0$ are forbidden by the
results of \ci{KT}, while
$\lam_j\le 0$ are
absent since the operator is nonnegative definite.

Denote by $S$  the circle $\{\theta\in R^2:|\theta|=1\}$ and by
$X$ the Hilbert space $L^2(S)$.
Since $L^2(\R^{2})=\cH_{ac}(\Lam(x_1))$ then
there exists a unitary generalized Fourier transform
\be\la{FF}
F(x_1): L^2(\R^{2})\to L^2([0,\infty),d\lam,X)
\ee
such that  functions $\psi\in C_0^\infty(\R^2) $ and
the operator $\Lam(x_1)$ admit the spectral
representations
\be\la{Msr}
F(x_1):\psi(x_1,\cdot) \mapsto \ti\psi(x_1,\lam)~;
~~~~~~
F(x_1)[\Lam(x_1)\psi(x_1,\cdot)] = \lam\ti\psi(x_1,\lam),~~\lam\ge 0.
\ee
The transform
exists by \ci[Theorem 4.2]{Iw82} since  $\cH_{sing}(\Lam(x_1))=0$ and 
$\cH_{pp}(\Lam(x_1))=0$.

Now (\re{A2ecrs}) is equivalent  to bound
\be\la{A2ecrr}
 \int\langle x_1\rangle^{-2\si}\Vert\na_1  \ti\psi(x_1,\lam)
 \Vert_X^2~d\lam dx_1\le C(\si)\int\langle x_1\rangle^{2\si}
 \Vert (-\na_1^2 +\lam-\om)\ti\psi(x_1,\lam)\Vert_X^2~d\lam dx_1.
\ee
Finally, (\re{A2ecrr}) follows
by the Fubini theorem from vector valued version of \ci[Lemma A.2]{A}
(see also \ci[Lemma 4, p. 442]{RS}).
\end{proof}

Next lemma generalizes Lemma A.3 from \ci{A}.
\begin{lemma}\la{Agm3}
For any $s\in\R$,  $b>0$, and $\psi\in \cD$,
the  estimate holds
\be\la{A3}
  \Vert \psi \Vert_{\cH^l_s}\le C(s,b)|\om|^{-\frac{1-l}2}
  \Big(\Vert(H_A-\om)\psi \Vert_{\cH^0_s}
  +\Vert(-i\na-A) \psi \Vert_{\cH^0_s}\Big),
  \quad \om\in\C,~~|\om|\ge b,\quad l=0,1.
\ee
\end{lemma}
\begin{proof}
{\it Step i)}
First consider $s=0$.
By (\re{eqv1}),
in this case (\re {A3}) is equivalent to estimate
\be\la{A3e}
 \Vert \psi \Vert^2+\Vert H_A^{l/2}\psi \Vert^2\le C(b)|\om|^{-({1-l})}
 \Big(\Vert(H_A-\om)\psi \Vert^2+\Vert H_A^{1/2}\psi\Vert^2\Big),
 \quad l=0,1.
\ee
We will deduce (\re{A3e}) from bound
\be\la{A4}
(1+\lam^{l/2})^2\le C(b)|\om|^{-(1-l)} (|\lam-\om|^2+\lam),
\quad \lam\ge 0,\quad \om\in\C,~~~|\om|\ge b,~~~l=0,1.
\ee
In the case $l=1$ the bound is trivial, and in the case
$l=0$ it is evident   separately for $|\lam-\om|<|\om|/2$
and for  $|\lam-\om|>|\om|/2$.

For the selfadjoint operator $H_A$ and functions $\psi\in \cD$
the spectral representations of type (\re{Msr}) also hold:
$$
F:\psi \mapsto \ti\psi(\lam),
~~~~~~
F: H_A\psi \mapsto \lam\ti\psi(\lam)
$$
Here $F: L^2(\R^{3})\to L^2([0,\infty),d\lam,Y)$ is a unitary operator,
$Y=L^2(S)$, and $S$ is the sphere $\{\theta\in\R^3: |\theta|=1\}$.

Multiplying both sides of (\ref{A4}) by $\Vert\ti\psi(\lam)\Vert_Y^2$
and integrating in $\lam$, we obtain
(\re{A3e}).
\medskip

{\it Step ii)}
Now consider an arbitrary $s\in\R$.
For an $\ve>0$,
denote $\rho_\ve(x)=(1+|\ve x|^2)^{1/2}$.
Then (\ref{A3}) is equivalent to the estimate:
\be\la{A3a}
\sum_{|\al|\le l}\Vert\rho_\ve^s\pa^\al\psi\Vert
\le C(s,b)|\om|^{-\frac{1-l}2}
  \Big(\Vert\rho_\ve^s(H_A-\om)\psi \Vert
  +\Vert\rho_\ve^s
(-i\na-A) \psi \Vert\Big),~~~~~~|\om|\ge b,
\ee
since for any fixed $\ve >0$
the weighted norm with  $\rho_{\ve}(x)$
is equivalent to the weigted norm with $\rho_{1}(x)$ defined in (\re{norm}).
We apply (\ref{A3e}) to $\rho_\ve^s(x)\psi(x)$ and obtain
\be\la{A5a}
\sum_{|\al|\le l}\Vert\pa^\al[\rho_\ve^s\psi]\Vert
\le C(b)|\om|^{-\frac{1-l}2}
 \Big(\Vert(H_A-\om)[\rho_\ve^s\psi] \Vert+\Vert (-i\na-A)[\rho_\ve^s
\psi]\Vert\Big),\quad l=0,1.
\ee
To deduce (\ref{A3a}) from (\ref{A5a}),
we should commute the multiplicators $\rho_\ve$ with  differential
operators. For example, consider the commutators
\be\la{comm}
  \pa^{\al}(\rho_{\ve}^s\psi )-\rho_{\ve}^s\pa^{\al}\psi
  =\sum\limits_{0\le\gamma_j\le\al_j,~|\gamma|~\ge 1}
  C_{\al,\gamma}\pa^{\gamma}\rho_{\ve}^s\cdot \pa^{\al-\gamma}\psi.
\ee
The commutators are
small and their contributions are negligible for small $\ve$.
Namely,
$$
 |{\na_j}\rho_\ve^s(x)|=\big|\frac s2(1+|\ve x|^2)^{s/2-1}2\ve^2x_j\big|
 \le \frac{|s|}2(1+|\ve x|^2)^{s/2-1}\ve(1+\ve^2x_j^2)\le \ve C\rho_\ve^s(x)
$$
where $C=C(s)$. Similarly, we have
$$
 |\pa^{\al}\rho_{\ve}^s(x)|\le \ve^{|\al|}C\rho_{\ve}^s(x),
 \quad x\in\R^3,\quad 0\le|\al|\le 2.
$$
Hence, (\ref{comm}) implies that
\be\la{A7}
  \Vert \pa^{\al}(\rho_{\ve}^s\psi )-\rho_{\ve}^s\pa^{\al}\psi \Vert
  \le\ve C_1\sum\limits_{|\gamma|\le|\al|-1}\Vert\rho_\ve^s\pa^\gamma\psi\Vert.
\ee
Therefore,
\beqn\la{A81}
  \Vert(H_A-\om)(\rho_{\ve}^s\psi )-\rho_{\ve}^s(H_A-\om)\psi \Vert
  &\le&\ve C_2
  \sum\limits_{|\gamma|\le 1}\Vert\rho_\ve^s \pa^\gamma \psi \Vert,
  \\
  \Vert(-i\na-A)(\rho_{\ve}^s\psi )-\rho_{\ve}^s(-i\na-A)\psi \Vert
  &\le&\ve C_3\Vert\rho_\ve^s \psi \Vert\la{A8}.
\eeqn
\medskip

{\it Step iii)} Now we can prove (\ref{A3a}).
First, we prove it for $l=0$.
Applying  (\ref{A5a}),
we obtain by (\ref{A81}) and (\ref{A8}) that
\beqn\nonumber
  \Vert \rho_{\ve}^s\psi \Vert&\le& C(b)
  \frac 1{\sqrt{|\om|}}\Big(\Vert(H_A-\om)(\rho_{\ve}^s\psi )\Vert
  +\Vert {(-i\na-A)}(\rho_{\ve}^s\psi )\Vert\Big)\\
  \nonumber
  &\le& C(b)\frac 1{\sqrt{|\om|}}\Big(\Vert\rho_{\ve}^s
  (H_A-\om)\psi \Vert+\ve C_2(\Vert \rho_{\ve}^s\psi \Vert
  +\Vert \rho_{\ve}^s{(-i\na-A)}\psi \Vert)\\
  \nonumber
  &+&\Vert \rho_{\ve}^s{(-i\na-A)}\psi \Vert
  +\ve C_1\Vert \rho_{\ve}^s\psi \Vert\Big)\\
  \nonumber
  &\le& C(b)\fr 1{\sqrt{|\om|}}\Big(\Vert\rho_{\ve}^s
  (H_A-\om)\psi \Vert
  +\Vert \rho_{\ve}^s{(-i\na-A)}\psi \Vert\Big)
  +\frac 1{\sqrt \om}\ve C_1(b)\Vert \rho_{\ve}^s\psi \Vert.
\eeqn
Choosing $\ve>0$ small enough such that
$\ve C_1(b)/\sqrt{|\om|}<1$, we obtain
$$
  \Vert \rho_{\ve}^s\psi \Vert\le C_2(b)\frac 1{\sqrt{|\om|}}
  \Big(\Vert\rho_{\ve}^s(H_A-\om)\psi \Vert
  +\Vert \rho_{\ve}^s(-i\na-A)\psi \Vert\Big).
$$
Hence,
(\ref{A3a}) with $l=0$ follows. \\
Finally, we prove (\ref{A3a}) for $l=1$.
Applying (\ref{A5a}), we obtain by
(\ref{A7}) with $|\al|= 1$ and (\ref{A81}), (\ref{A8}),
that
\beqn\nonumber
  \sum\limits_{|\al|\le 1}\Vert \rho_{\ve}^s\pa^{\al}\psi \Vert
  &\le&\sum\limits_{|\al|\le 1}\Vert \pa^{\al}(\rho_{\ve}^s\psi )\Vert
  +\ve C_1\Vert\rho_\ve^s \psi \Vert\\
  \nonumber
  &\le& C(b)\Big(\Vert(H_A-\om)(\rho_{\ve}^s\psi )\Vert
  +\Vert {(-i\na-A)}(\rho_{\ve}^s\psi )\Vert\Big)
  +\ve C_1\Vert\rho_\ve^s \psi \Vert\\
  \nonumber
  &\le& C(b)\Big(\Vert\rho_{\ve}^s(H_A-\om)\psi \Vert
  +\Vert \rho_{\ve}^s{(-i\na-A)}\psi \Vert\Big)
  +\ve C_3(b)
 \sum\limits_{|\ga|\le 1}\Vert \rho_{\ve}^s\pa^{\ga}\psi \Vert.
\eeqn
Choosing $\ve>0$ small enough, we obtain
$$
 \sum\limits_{|\al|\le 1}\Vert \rho_{\ve}^s\pa^{\al}\psi \Vert
 \le C_4(b)\Big(\Vert\rho_{\ve}^s(H_A-\om)\psi \Vert
 +\Vert \rho_{\ve}^s(-i\na-A)\psi \Vert\Big)
$$
that implies (\ref{A3a}) with $l=1$.
Lemma \re{Agm3} is proved.
\end{proof}

{\bf Proof of Theorem \ref{AJK}}
Combining (\ref{A3}) with $s=-\si$ and (\ref{A2}), we obtain
for all $\psi \in {\cal D}$
\beqn\nonumber
  \Vert \psi \Vert_{\cH^l_{-\si}}
  &\le& C(\si,b)|\om|^{-\fr{1-l}2}\Big(\Vert (H_A-\om)\psi
  \Vert_{\cH^0_{-\si}}
  +C(\si)\Vert (H_A-\om)\psi \Vert_{\cH^0_\si}\Big)\\
  \nonumber
  &\le& C_1(\si,b)|\om|^{-\fr{1-l}2}\Vert (H_A-\om)\psi \Vert_{\cH^0_\si}
\eeqn
and then Theorem \re{AJK} is proved.


\end{document}